\documentclass[11pt,a4paper]{article}
\usepackage[T1]{fontenc}
\usepackage[utf8]{inputenc}
\usepackage{lmodern,microtype}
\usepackage[margin=25mm]{geometry}
\usepackage{amsmath,amssymb,amsthm,mathtools,bm}
\usepackage{graphicx,booktabs,authblk,cite,placeins}
\usepackage[hidelinks]{hyperref}
\graphicspath{{figures/}}
\allowdisplaybreaks[1]
\newtheorem{theorem}{Theorem}[section]
\newtheorem{lemma}[theorem]{Lemma}
\newtheorem{proposition}[theorem]{Proposition}
\newtheorem{corollary}[theorem]{Corollary}
\theoremstyle{remark}
\newtheorem{remark}[theorem]{Remark}
\newcommand{\m}{\bm m}
\newcommand{\n}{\bm n}
\newcommand{\Sop}{\widehat{\bm S}}
\newcommand{\K}{\mathcal K}
\newcommand{\T}{\mathcal T}
\newcommand{\V}{\mathcal V}
\newcommand{\Qsites}{\mathcal Q}
\newcommand{\Fsites}{\mathcal F}
\newcommand{\R}{\mathbb R}
\newcommand{\C}{\mathbb C}
\newcommand{\ket}[1]{\lvert#1\rangle}
\newcommand{\bra}[1]{\langle#1\rvert}
\newcommand{\norm}[1]{\lVert#1\rVert}
\newcommand{\abs}[1]{\lvert#1\rvert}
\DeclareMathOperator{\Tr}{Tr}
\DeclareMathOperator{\conv}{conv}
\DeclareMathOperator{\cone}{cone}
\DeclareMathOperator{\dist}{dist}
\DeclareMathOperator{\Arg}{Arg}
\newcommand{\doi}[1]{\href{https://doi.org/#1}{\nolinkurl{#1}}}
\hypersetup{pdftitle={Exact admissibility radii for the lattice topological charge
of quantum spin textures},
 pdfauthor={Raul Sanchez Galan and Robert Wieser},
 pdfkeywords={quantum skyrmions, lattice degree, convex geometry, spectral gap, verified computation}}

\title{Exact admissibility radii for the lattice topological charge
of quantum spin textures}
\author[1]{Ra\'ul S\'anchez Gal\'an}
\author[2]{Robert Wieser}
\affil[1]{\small 4i Intelligent Insights, 41092 Sevilla, Spain}
\affil[2]{\small School of Physics and Optoelectronic Engineering,
Nanjing University of Information Science and Technology,
Nanjing 210044, China}
\date{}

\begin{document}
\maketitle
\begin{abstract}
How much can a magnetic field change before a quantum spin texture
loses its skyrmion charge? For a triangulated spin-$s$ texture on $N$
quantum sites, we determine the smallest change in its local moments
that makes the reconstructed topological charge ill-defined, allowing
unequal error bounds and fixed boundary moments. Combined with spectral
perturbation theory, this geometric threshold yields explicit
magnetic-field intervals that preserve the ground-state charge, scaling
as $1/N$ when the gap and geometric margin have positive limits. We
apply the bounds to finite spin-$1/2$ systems with exchange and
Dzyaloshinskii--Moriya interactions. In a rigorously verified seven-spin
example the charge changes at a single certified field, while the
excitation gap and all local spin polarizations remain bounded away from
zero; there the geometric margin vanishes linearly.
\end{abstract}

\section{Introduction}

Quantum magnetic skyrmions can be studied through local spin moments,
correlation functions and many-body wave functions
\cite{lohani,siegl,haller2022,haller2024}. These descriptions carry
different information. In particular, normalizing local expectation
values and computing their winding produces an invariant of a
reconstructed magnetization field on a chosen mesh. The construction
does not define a topological invariant of the unrestricted quantum-state
space. Sotnikov et al.\ \cite{sotnikov} show that in a translationally
invariant cluster the local magnetization is uniform, so that the
reconstruction carries no information there, and propose the scalar
chirality, a three-spin correlation function, as a many-body diagnostic;
Salvati et al.\ \cite{salvati} find that this chirality and the spin
structure factor are largely unaffected by local projective measurements,
despite the absence of topological protection. The prescribed exterior
used below breaks translation invariance, so the reconstruction retains
content in the present geometry. Local moments can become small, and a
spherical interpolation can become singular even if each moment remains
nonzero.

The degree of Berg and L\"uscher \cite{berg} makes the reconstruction
precise on a triangulated surface. Its regularity is equivalent to the
three directions on each face lying in a common open hemisphere.
Quantum-skyrmion studies already distinguish normalized winding from
unnormalized spin diagnostics. Siegl et al.\ \cite{siegl} show that
the reconstructed charge can change without an energy-level
crossing when a local moment vanishes and reverses.
Our verified example instead keeps every local polarization bounded
away from zero and isolates a singular face interpolation.
The quantitative question is how large local-data or Hamiltonian
errors can be before the degree ceases to be certified.

Related geometric restrictions already occur in classical lattice
models. Ward \cite{ward} enforces acute angles between neighbouring
spins through a modified energy to prevent unwinding. Briani, Cicalese
and Kreutz \cite{briani} construct discrete charges and refined
interpolations for a lattice-to-continuum analysis that also handles
ambiguous configurations. The distance-to-failure viewpoint itself is
established in convex feasibility theory, beginning with Renegar's
distance to ill-posedness \cite{renegar} and including block
perturbations \cite{canovas}. Here we quantify additive errors in a
specified moment array and use that distance to certify a quantum
ground-state texture.

We give an exact answer in local-moment space. Each face contributes a
convex optimization problem whose distance from zero is the precise
radius of admissibility. The formulation applies to arbitrary spin $s$
and accommodates nonuniform error budgets and prescribed boundary moments.
Fixing the exterior matters:
allowing fictitious errors in the boundary can unnecessarily reduce a
certificate. We then combine this geometry with standard eigenvector
perturbation estimates \cite{davis} to obtain explicit magnetic-field
intervals and an a posteriori criterion for computed ground states.
Our contribution is to combine the Berg--L\"uscher degree and
Davis--Kahan-type state bounds into an explicit, computable magnetic-field
certificate, through an exact weighted distance to inadmissible moments
and lower bounds checkable in rational arithmetic. We also identify the
assumptions under which the field certificate scales as $1/N$, and show
that using the sharp Davis--Kahan $\sin2\Theta$ bound \cite{davis} in
place of the $\sin\Theta$ estimate widens every such interval by a factor
between $1$ and $3\sqrt3/4$ under the same hypotheses.

The physical example is a chiral spin-$1/2$ Heisenberg model with a
prescribed polarized environment, following the finite-flake setting
of \cite{haller2024}. Related work studies a quantum--classical
interpolation in finite chiral spin rings
\cite{SanchezGalanWieser2026}. Here all Hamiltonian coefficients are
independent of the quantum state. Seven-spin computations are verified by integer and
rational arithmetic, including the spectral ordering and the integer
degree. They establish a charge change while both the gap and all local
polarizations have positive uniform lower bounds. Full-Hilbert-space
nineteen-spin calculations provide a larger numerical illustration.
The claims concern finite systems and reconstruction on the specified
mesh; no thermodynamic phase boundary is inferred.

\section{An exact local-data radius}
\label{sec:geometry}

\subsection{Admissibility and degree}

Let $\K$ triangulate a closed oriented surface $\Sigma$, with vertex set
$\V$ and oriented triangular faces $\T$. Write
$\V=\Qsites\sqcup\Fsites$, where $\Qsites$ consists of $N$ quantum
spin-$s$ sites, $s\in\{1/2,1,3/2,\ldots\}$, and $\Fsites$ contains
prescribed classical data, if any. Set $\hbar=1$ and let $\Sop_i$ be
the spin-$s$ generators on $\C^{2s+1}$. For a density operator $\rho$,
\begin{equation}
 \m_i(\rho)=\Tr(\rho\Sop_i),\qquad
 \abs{\m_i}\leq s,\qquad p_i=\abs{\m_i}/s\in[0,1],
 \qquad i\in\Qsites.
 \label{eq:moments}
\end{equation}
The bound follows from $\norm{u\cdot\Sop_i}=s$ for every unit $u$, where $\norm\cdot$ denotes the operator norm for operators and the
Hilbert-space norm for state vectors; $\abs\cdot$ denotes the
Euclidean norm on $\R^3$.

The moments at $\Fsites$ are fixed and nonzero. Initially, all vectors
may simply be regarded as data in $\R^3$.

Call an array $\m$ \emph{admissible} if
\begin{equation}
 0\notin\conv\{\m_i,\m_j,\m_k\}
 \quad\text{for every }\tau=(i,j,k)\in\T.
 \label{eq:admissible}
\end{equation}
If $q\ne0$ is the closest point of a face hull to zero, the projection
inequality gives $q\cdot\m_r\geq\abs q^2>0$ for each vertex of that
face. Conversely, a strictly separating vector excludes zero from the
hull. Thus \eqref{eq:admissible} is equivalent to nonzero moments and
a common open hemisphere for the normalized directions
$\n_i=\m_i/\abs{\m_i}$ on each face. Related restrictions ensuring a well-defined, locally constant
topological charge occur in lattice gauge theory, notably in
L\"uscher's construction \cite{luscher1982}.
Here \eqref{eq:admissible} is necessary and sufficient for the
specific radial interpolation \eqref{eq:radial} to be well defined.
The weights introduced below specify the allowed perturbations of
the moment data.

For barycentric coordinates $\lambda_r\geq0$, $\sum_r\lambda_r=1$,
the radial interpolation
\begin{equation}
 F_{\m}\big|_\tau(\lambda)=
 \frac{\sum_{r\in\tau}\lambda_r\m_r}
      {\abs{\sum_{r\in\tau}\lambda_r\m_r}}
 \label{eq:radial}
\end{equation}
is defined exactly on the admissible set. The face maps agree on common
edges and form a continuous map $F_{\m}:\Sigma\to S^2$. With
\begin{equation}
 a_\tau=1+\n_i\cdot\n_j+\n_j\cdot\n_k+\n_k\cdot\n_i,
 \qquad b_\tau=\n_i\cdot(\n_j\times\n_k),
 \label{eq:ab}
\end{equation}
the oriented spherical area is $2\Arg(a_\tau+\mathrm i b_\tau)$,
using the principal argument. Hence \cite{berg}
\begin{equation}
 Q_\K(\m)=\frac1{2\pi}\sum_{\tau\in\T}
 \Arg(a_\tau+\mathrm i b_\tau)=\deg F_{\m}\in\mathbb Z.
 \label{eq:degree}
\end{equation}
Each face image lies in an open hemisphere, so the formula uses an
unambiguous area in $(-2\pi,2\pi)$, including the continuous value on
degenerate admissible faces. A continuous admissible path gives a
homotopy of \eqref{eq:radial}, and its degree is constant.

\subsection{Unequal errors and fixed data}

Assign weights $w_i\geq0$, at least one of which is positive. A vertex
with $w_i=0$ is fixed. In the affine space of arrays sharing these fixed
values, define
\begin{equation}
 d_w(\m,\widetilde\m)=
 \max_{i:w_i>0}\frac{\abs{\m_i-\widetilde\m_i}}{w_i}.
 \label{eq:weightednorm}
\end{equation}
For a face containing a variable vertex, set
\begin{align}
 P_\tau(\m,w)
 &=\left\{\sum_{i\in\tau}\lambda_i\m_i:
        \lambda_i\geq0,\ \sum_{i\in\tau}\lambda_iw_i=1\right\}
 \label{eq:polyhedron}\\
 &=\conv\{\m_i/w_i:i\in\tau,\ w_i>0\}
   +\cone\{\m_i:i\in\tau,\ w_i=0\}.
 \nonumber
\end{align}
This is a nonempty closed polyhedron. Define
\begin{equation}
 r_{w,\tau}(\m)=\dist(0,P_\tau(\m,w)),\qquad
 r_w(\m)=\min_{\tau\in\T}r_{w,\tau}(\m),
 \label{eq:weightedradius}
\end{equation}
assigning $r_{w,\tau}=+\infty$ to wholly fixed faces. The fixed data
are assumed compatible with an admissible array; in particular, zero
is not a convex combination of fixed vertices within any face.

\begin{theorem}[Exact weighted radius]
\label{thm:radius}
Let $\mathcal B_w$ be the inadmissible arrays with the prescribed fixed
values. Then $r_w$ is $1$-Lipschitz in $d_w$ and
\begin{equation}
 \dist_{d_w}(\m,\mathcal B_w)=r_w(\m).
 \label{eq:sharp}
\end{equation}
For admissible $\m$, if $d_w(\m,\widetilde\m)<r_w(\m)$, the straight
interpolation remains admissible and
$Q_\K(\widetilde\m)=Q_\K(\m)$. The distance equality also holds
when the original and perturbed variable moments are constrained by
$\abs{\m_i}\leq s$, the physical spin-$s$ moment balls.
\end{theorem}
\begin{proof}
For coefficients satisfying \eqref{eq:polyhedron},
\begin{equation}
 \left|\sum_i\lambda_i(\widetilde\m_i-\m_i)\right|
 \leq d_w(\m,\widetilde\m)\sum_i\lambda_iw_i
 =d_w(\m,\widetilde\m).
 \label{eq:weightedchange}
\end{equation}
Taking infima and exchanging the arrays proves the Lipschitz claim.
An inadmissible face has a zero convex combination involving a variable
vertex; rescaling its coefficients makes $\sum_i\lambda_iw_i=1$.
Conversely, a zero combination in \eqref{eq:polyhedron} can be normalized
to a convex combination. Thus $r_w=0$ exactly on $\mathcal B_w$, which
proves the lower bound in \eqref{eq:sharp}.

For the reverse bound, choose a minimizing face and its closest point
$q\in P_\tau$. Replace $\m_i$ on this face by $\m_i-w_iq$, leaving
all other vertices unchanged. If $q=\sum_i\lambda_i\m_i$ with
$\sum_i\lambda_iw_i=1$, the same combination of the changed moments
vanishes. The new array is inadmissible and its distance is $\abs q$.
For each variable vertex, $\m_i/w_i\in P_\tau$, so projection gives
$q\cdot\m_i\geq w_i\abs q^2$. Consequently,
\begin{equation}
 \abs{\m_i-w_iq}^2
 \leq\abs{\m_i}^2-w_i^2\abs q^2\leq s^2
 \label{eq:blochsharp}
\end{equation}
because the original spin-$s$ moments satisfy
$\abs{\m_i}\leq s$ by \eqref{eq:moments}, and
$w_i^2\abs q^2\geq0$. The attaining perturbation respects the physical
one-site bounds. To see that these bounds characterize all attainable
first moments, let $\ket{s,n}$ be the eigenstate of $n\cdot\Sop$
with eigenvalue $s$. A vector $m=rn$, $0\leq r\leq s$, is realized by
\begin{equation}
 \rho_m=\frac{1+r/s}{2}\ket{s,n}\bra{s,n}
       +\frac{1-r/s}{2}\ket{s,-n}\bra{s,-n}.
 \label{eq:momentrealization}
\end{equation}
For $r=0$ any unit $n$ may be chosen. The tensor product of these
one-site states realizes the entire changed array. Thus the sharpness
statement holds for every spin $s$, without a pure-state restriction.

Finally, the Lipschitz bound along the straight interpolation gives
$r_w(\m(t))\geq r_w(\m)-t d_w(\m,\widetilde\m)>0$.
Homotopy invariance proves the charge statement.
\end{proof}

\begin{figure}[tbp]
 \centering
 \includegraphics[width=0.95\textwidth]{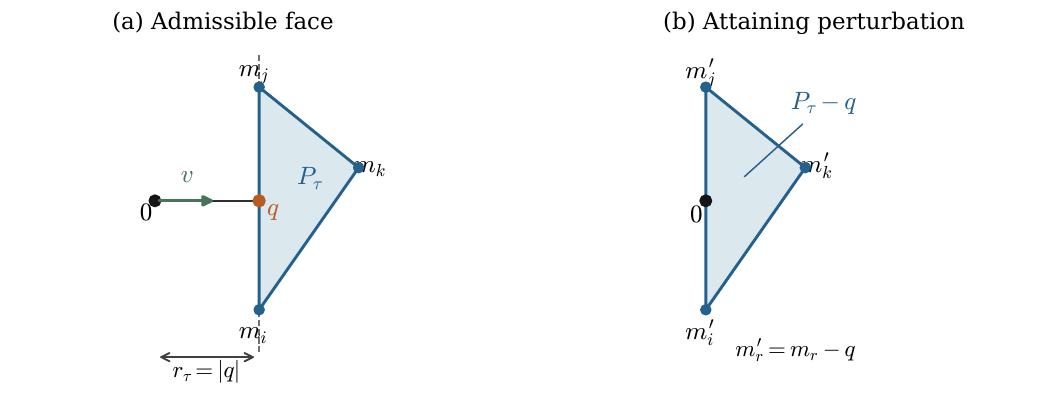}
 \caption{Planar schematic of the projection construction when all face
 weights equal one. (a) The face hull $P_\tau$ excludes the origin. Its
 closest point $q$ determines the radius $r_\tau=\abs q$ and a unit
 separating vector $v=q/\abs q$, with $v\cdot x\geq r_\tau$ for
 $x\in P_\tau$. (b) Translating all three moments by $-q$ makes the
 hull contain zero, attaining the distance to inadmissibility. The same
 construction holds in $\R^3$; unequal weights and fixed data replace
 the hull by \eqref{eq:polyhedron} and the shifts by $-w_iq$.}
 \label{fig:admissibility}
\end{figure}

If every weight is one, \eqref{eq:weightedradius} reduces to
\begin{equation}
 \gamma(\m)=\min_{\tau\in\T}
 \dist\bigl(0,\conv\{\m_i:i\in\tau\}\bigr).
 \label{eq:gamma}
\end{equation}
For quantum sites with a fixed exterior, we instead use
\begin{equation}
 w_i=1\ (i\in\Qsites),\quad w_i=0\ (i\in\Fsites),
 \qquad \Gamma(\m)=r_w(\m).
 \label{eq:Gamma}
\end{equation}
Here $\Gamma\geq\gamma$, because the allowed perturbations have been
restricted. Moreover $\Gamma\leq\min_{i\in\Qsites}\abs{\m_i}\leq s$.
The exactness in Theorem~\ref{thm:radius} concerns the distance to
\emph{inadmissible local data}; it is not an optimal distance to another
charge sector or an optimal perturbation of a prescribed pure many-body
state.
Figure~\ref{fig:admissibility} illustrates the attaining perturbation.
For classical spins constrained to retain a fixed length, the radius
still gives a sufficient ambient error bound; sharpness is not asserted
under that additional constraint.

\subsection{Computable lower bounds}

For a face containing a variable vertex, the dual form of the radius is
\begin{equation}
 r_{w,\tau}=
 \max_{\substack{\abs v\leq1\\v\cdot\m_j\geq0\ (w_j=0)}}
 \min_{i\in\tau:w_i>0}\frac{v\cdot\m_i}{w_i}.
 \label{eq:dual}
\end{equation}
This is the Euclidean-distance form of separating-hyperplane duality;
Farkas' lemma gives the underlying alternative between $0\in P_\tau$
and a feasible $v$ with positive margin \cite[sections~2.5 and~5.8]{boyd}.
Indeed, every feasible $v$ bounds the distance from below, by applying
Cauchy--Schwarz to the representation \eqref{eq:polyhedron}. For the
closest point $q\ne0$, projection onto the polyhedron gives
$q\cdot\m_j\geq0$ along each fixed-vertex cone generator and
$q\cdot\m_i/w_i\geq\abs q^2$ for each variable vertex. Thus
$v=q/\abs q$ attains the lower bound. If $q=0$, use $v=0$.

For $r_{w,\tau}>0$, the optimal separating vector is unique. Indeed,
any optimizer $v$ satisfies
\[
 r_{w,\tau}\leq v\cdot q\leq\abs v\,\abs q\leq r_{w,\tau},
\]
so equality in Cauchy--Schwarz forces $v=q/\abs q$.
At $r_{w,\tau}=0$, the maximizers instead form the set
$\{v:\abs v\leq1,\ v\cdot\m_i\geq0\text{ for all }i\in\tau\}$.
For a triangular face they are nonunique: $0\in P_\tau$ implies
linear dependence of its three moment vectors, whose common orthogonal
complement supplies nonzero maximizers as well as $v=0$.
The minimizing face and the coefficients representing $q$ need not be
unique even when its separating vector is unique.

A feasible point in $P_\tau$ gives an upper bound, while a feasible
separating vector in \eqref{eq:dual} gives a lower bound. This distinction
permits rigorous certification with rational vectors without requiring
the exact optimizer. In the polarized
exterior used below, the cone is generated by $\widehat{\bm z}$, and
the additional dual constraint is simply $v_z\geq0$.

\section{Certificates for quantum ground states}
\label{sec:spectral}

\subsection{From state errors to moment errors}

For density operators $\rho$ and $\omega$, write
$D(\rho,\omega)=\tfrac12\norm{\rho-\omega}_1$, and let $\rho_i$ and
$\omega_i$ denote their reduced states at site $i$. Duality of the
trace and operator norms, followed by contraction under partial trace,
gives
\begin{equation}
 \max_{i\in\Qsites}\abs{\m_i(\rho)-\m_i(\omega)}
 \leq 2s\max_iD(\rho_i,\omega_i)\leq 2sD(\rho,\omega).
 \label{eq:contraction}
\end{equation}
Indeed, the moment difference is the supremum of
$\abs{\Tr((\rho_i-\omega_i)u\cdot\Sop_i)}$ over unit $u$, and
$\norm{u\cdot\Sop_i}=s$. The factor $2s$ is attained by opposite
fully polarized one-site states. For $s=1/2$, the first inequality is
an equality because $\rho_i=I/2+\m_i\cdot\bm\sigma$, where $\bm\sigma$
is the vector of Pauli matrices.
Consequently, if $2sD(\rho,\omega)<\Gamma(\rho)$, their reconstructed
charges agree and
$\Gamma(\omega)\geq\Gamma(\rho)-2sD(\rho,\omega)$.
For pure states,
\begin{equation}
 D(\psi,\phi)=\sqrt{1-\abs{\langle\psi|\phi\rangle}^2}.
 \label{eq:puredistance}
\end{equation}
This provides a sufficient global-state bound from the sharp local-data
radius.

\subsection{A common spectral estimate}

Both perturbation and a posteriori certificates follow from one
spectral-residual argument.

\begin{lemma}[Spectral separation and texture stability]
\label{lem:transfer}
Let $H$ be Hermitian with a simple normalized ground state $\psi_0$, and let $\phi$
be normalized. Write $E_0(H)\leq E_1(H)\leq\cdots$ for its ordered
eigenvalues. Choose $\lambda\in\R$ and $\delta>0$ with
$\delta\leq E_1(H)-\lambda$, and set
$\eta=\norm{(H-\lambda)\phi}/\delta$. Then
\begin{equation}
 D(\psi_0,\phi)\leq\eta,\qquad
 \abs{\Gamma(\psi_0)-\Gamma(\phi)}\leq2s\eta.
 \label{eq:masterbound}
\end{equation}
If either state has geometric radius greater than $2s\eta$, both
reconstructions are admissible and their degrees agree.
\end{lemma}
\begin{proof}
For $P_\perp=I-\ket{\psi_0}\bra{\psi_0}$, spectral separation gives
\begin{equation}
 \delta\norm{P_\perp\phi}
 \leq\norm{P_\perp(H-\lambda)\phi}
 \leq\norm{(H-\lambda)\phi}.
 \label{eq:projected}
\end{equation}
Now $D(\psi_0,\phi)=\norm{P_\perp\phi}$. Equation~\eqref{eq:contraction}
and the Lipschitz and homotopy statements of Theorem~\ref{thm:radius}
give the remaining claims. This is the elementary spectral-separation
mechanism underlying eigenvector-rotation bounds \cite{davis}.
\end{proof}

\subsection{Hamiltonian and residual corollaries}

\begin{corollary}[Ground-state perturbation certificate]
\label{thm:perturbation}
Let $H_0$ be Hermitian, with simple ground state $\psi_0$, gap
$\Delta_0=E_1(H_0)-E_0(H_0)>0$ and admissible moments of radius
$\Gamma_0$. For a Hermitian perturbation $V$, set $\varepsilon=\norm V$.
If $\varepsilon<\Delta_0/2$, then $H_0+V$ has a simple ground state
$\psi_V$, gap at least $\Delta_0-2\varepsilon$, and
\begin{equation}
 D(\psi_0,\psi_V)\leq\frac{\varepsilon}{\Delta_0-\varepsilon}.
 \label{eq:rotation}
\end{equation}
In particular,
\begin{equation}
 \varepsilon<\frac{\Delta_0\Gamma_0}{2s+\Gamma_0}
 \quad\Longrightarrow\quad Q_\K(\psi_V)=Q_\K(\psi_0).
 \label{eq:operatorcertificate}
\end{equation}
\end{corollary}
\begin{proof}
Writing $E_V=E_0(H_0+V)$, the min--max principle gives the gap bound and
$E_V\leq E_0(H_0)+\varepsilon$. Apply Lemma~\ref{lem:transfer} with
$H=H_0$, $\phi=\psi_V$, $\lambda=E_V$ and
$\delta=\Delta_0-\varepsilon$, since
$\norm{(H_0-E_V)\psi_V}=\norm{V\psi_V}\leq\varepsilon$.
Condition~\eqref{eq:operatorcertificate} makes
$2s\varepsilon/(\Delta_0-\varepsilon)<\Gamma_0$; as $\Gamma_0\leq s$,
it also ensures $\varepsilon<\Delta_0/2$.
\end{proof}

\begin{corollary}[Verified residual certificate]
\label{prop:residual}
Suppose $H$ has a simple ground state $\psi_0$ and a verified bound
$b\leq E_1(H)$. For a normalized candidate $\phi$ and a real
$\lambda<b$, let $r=\norm{(H-\lambda)\phi}$ and
$\eta=r/(b-\lambda)$. If $2s\eta<\Gamma(\phi)$, then
\begin{equation}
 Q_\K(\psi_0)=Q_\K(\phi),\qquad
 \Gamma(\psi_0)\geq\Gamma(\phi)-2s\eta.
 \label{eq:aposteriori}
\end{equation}
\end{corollary}
\begin{proof}
Apply Lemma~\ref{lem:transfer} with $\delta=b-\lambda$.
\end{proof}

A small residual alone does not identify the ground state: the lower
bound on the first excited energy is essential. In the seven-spin
example, Appendix~\ref{app:verification} obtains this bound from a complete
approximate eigenbasis, whose residual and orthogonality error are
bounded with exact arithmetic.

\subsection{Magnetic-field radius and size scaling}

For $H(B)=H(B_0)-(B-B_0)\sum_{i\in\Qsites}\widehat S_i^z$, the field
parameter $B$ has energy units and
$\norm{H(B)-H(B_0)}=Ns\abs{B-B_0}$. Thus the explicit sufficient
field radius is
\begin{equation}
 \abs{B-B_0}<R_B(B_0),\qquad
 R_B(B_0)=\frac{\Delta_0\Gamma_0}{Ns(2s+\Gamma_0)}.
 \label{eq:fieldradiusspin}
\end{equation}
For $s=1/2$ this specializes to the formula used in the computations,
\begin{equation}
 R_B(B_0)=\frac{2\Delta_0\Gamma_0}{N(1+\Gamma_0)}.
 \label{eq:fieldradius}
\end{equation}
For independently perturbed local fields,
$V=-\sum_i\delta\bm b_i\cdot\Sop_i$, the commuting one-site terms
give $\norm V=s\sum_i\abs{\delta\bm b_i}$. Equation
\eqref{eq:operatorcertificate} then certifies a total local-field error
budget. Scalar energy shifts can always be subtracted from $V$ before
using the theorem.

\begin{proposition}[Size scaling of the certificate]
\label{prop:scaling}
Consider a sequence of systems at fixed spin $s$ with simple ground
states, gaps $\Delta_N>0$ and admissible moment radii $\Gamma_N>0$.
The uniform-field certificate satisfies the exact identity
\begin{equation}
 R_{B,N}=\frac1N\frac{\Delta_N\Gamma_N}{s(2s+\Gamma_N)}.
 \label{eq:Nidentity}
\end{equation}
If $0<\Delta_-\leq\Delta_N\leq\Delta_+<\infty$ and
$0<\Gamma_-\leq\Gamma_N\leq\Gamma_+\leq s$, with bounds independent
of $N$, then
\begin{equation}
 \frac{\Delta_-\Gamma_-}{s(2s+\Gamma_-)}\frac1N
 \leq R_{B,N}\leq
 \frac{\Delta_+\Gamma_+}{s(2s+\Gamma_+)}\frac1N.
 \label{eq:Nscaling}
\end{equation}
Thus $R_{B,N}=\Theta(N^{-1})$. If, more specifically,
$\Delta_N\to\Delta_*>0$ and $\Gamma_N\to\Gamma_*>0$, then
\begin{equation}
 R_{B,N}\sim\frac{\Delta_*\Gamma_*}{s(2s+\Gamma_*)}\,N^{-1}.
 \label{eq:Nasymptotic}
\end{equation}
\end{proposition}
\begin{proof}
The commuting operators $\widehat S_i^z$ have eigenvalues in $[-s,s]$,
and their sum attains $\pm Ns$ on fully polarized product states.
Hence its norm is exactly $Ns$, giving \eqref{eq:Nidentity} from
\eqref{eq:operatorcertificate}. The function
$g\mapsto g/(2s+g)$ is increasing for $g\geq0$, since its derivative
is $2s/(2s+g)^2>0$. Substituting the uniform bounds proves
\eqref{eq:Nscaling}; taking limits in $NR_{B,N}$ proves
\eqref{eq:Nasymptotic}.
\end{proof}

\subsection{A sharper field certificate}

Both the perturbation certificate and the field radius above admit a
sharper constant under the same hypotheses.

\begin{corollary}[Ground-state rotation bound (Davis--Kahan $\sin2\Theta$)]
\label{cor:sharperradius}
Let $H_0$ have a simple ground state $\psi_0$ and gap at least
$\Delta>0$. For a Hermitian perturbation $V$ with
$\varepsilon=\norm V<\Delta/2$, its perturbed ground state $\psi_V$
is simple and obeys
\begin{equation}
 D(\psi_0,\psi_V)\leq
 \sin\!\left(\frac12\arcsin\frac{2\varepsilon}{\Delta}\right).
 \label{eq:sharprotation}
\end{equation}
If $\Gamma>0$ is the initial fixed-exterior margin of a spin-$s$
texture, put $\alpha=\Gamma/(2s)$. Its charge is preserved whenever
\begin{equation}
 \norm V<\Delta \alpha\sqrt{1-\alpha^2}.
 \label{eq:sharpnormcond}
\end{equation}
In particular, for a uniform field, the sufficient radius is
\begin{equation}
 R_B^{\sharp}=\frac{\Delta\Gamma}{2Ns^2}
       \sqrt{1-\frac{\Gamma^2}{4s^2}}.
 \label{eq:sharpradius}
\end{equation}
Relative to the original certificate, evaluated with the same gap value
or the same certified lower bound $\Delta$,
\begin{equation}
 \frac{R_B^{\sharp}}{R_B}=(1+\alpha)\sqrt{1-\alpha^2},\qquad
 1<\frac{R_B^{\sharp}}{R_B}\leq\frac{3\sqrt3}{4}.
 \label{eq:improvement}
\end{equation}
\end{corollary}

\begin{proof}
Inequality \eqref{eq:sharprotation} is the Davis--Kahan $\sin2\Theta$
theorem \cite{davis} specialised to a simple ground state; we give a
self-contained elementary proof.
Weyl's inequality gives a positive gap along $H_0+tV$, $0\leq t\leq1$.
Write its normalized ground state, with a suitable phase, as
$\phi_t=\cos\theta_t\psi_0+\sin\theta_t\chi_t$,
where $\chi_t\perp\psi_0$, $\norm{\chi_t}=1$ and
$0\leq\theta_t\leq\pi/2$. Where $\theta_t>0$, the unit vector
$\xi_t=-\sin\theta_t\psi_0+\cos\theta_t\chi_t$ is orthogonal
to $\phi_t$. Consequently,
\[
 0=\langle\xi_t,(H_0+tV)\phi_t\rangle,
 \qquad
 \frac{\Delta}{2}\sin(2\theta_t)\leq t\varepsilon.
\]
At $\theta_t=0$ the latter inequality is automatic.
The continuous angle starts at zero. It cannot reach $\pi/4$, because
that would require $t\varepsilon\geq\Delta/2$. Thus it remains on
the small-angle branch, and
$\theta_1\leq\tfrac12\arcsin(2\varepsilon/\Delta)$.
The identity $D(\psi_0,\psi_V)=\sin\theta_1$ proves
\eqref{eq:sharprotation}.

Since $0<\Gamma\leq s$, we have $0<\alpha\leq1/2$.
The condition in \eqref{eq:sharpnormcond} makes the right-hand
side of \eqref{eq:sharprotation} strictly less than $\alpha$.
The moment change is therefore less than $2s\alpha=\Gamma$, so the exact
geometric-radius theorem preserves the degree. Substituting
$\norm V=Ns\abs{\delta B}$ proves \eqref{eq:sharpradius}.
Taking the ratio to the original formula gives
\eqref{eq:improvement}; the function
$(1+\alpha)\sqrt{1-\alpha^2}$ increases on $[0,1/2]$ from $1$ to
$3\sqrt3/4$.
\end{proof}

\begin{remark}[Sharp state estimate and verified numerical gain]
The bound \eqref{eq:sharprotation} is optimal under these norm and gap
data; optimality is due to Davis \cite[Theorem~5.1]{davis1963}, and
Seelmann \cite[Remark~2.9]{seelmann} gives the arcsine form together
with an explicit extremal pair.
This does not assert that the resulting physical-field interval is
optimal after conversion to local moments.

For the seven-spin reference at $B_0/J=0.2$ and
$h_x=h_y=0$, the exact lower bounds for $\Delta/J$ and
$\Gamma$ given in \eqref{eq:certifiedreference} imply
$R_B^\sharp/J>0.07451$, rather than the original
$R_B/J>0.05992$. The new bound is checked without evaluating a floating
square root: for the exact rational lower bounds
$\underline\Delta$ and $\underline\Gamma$, verify
\[
 \left(\frac{2\underline\Delta\,\underline\Gamma}{7}\right)^2
 \bigl(1-\underline\Gamma^2\bigr)>(0.07451)^2.
\]
The left-hand side is monotone in these lower bounds on the physical
range $0<\underline\Gamma\leq1/2$. The sufficient $1/N$ scaling of
Proposition~\ref{prop:scaling} remains unchanged under the
uniform-gap and uniform-margin assumptions used there.
\end{remark}

\subsection{Why a gap alone is insufficient}

For local commuting-projector Hamiltonians satisfying the
topological-order conditions of Bravyi, Hastings and Michalakis
\cite{bhm}, sufficiently weak local perturbations preserve an isolated
low-energy spectral band, with a perturbation-strength threshold
independent of system size. These conclusions require additional
structure beyond a spectral gap. The reconstructed degree considered
here is a function of the local moments on a chosen mesh, and the
perturbations certified above are measured in global operator norm.
A simple analytic example separates the spectral and geometric
conditions. On an oriented tetrahedral sphere, prescribe
\begin{equation}
 u_j(t)=\left(\sqrt{1-t^2}\cos\frac{2\pi j}{3},
 \sqrt{1-t^2}\sin\frac{2\pi j}{3},t\right),\quad j=0,1,2,
 \qquad u_3=-\widehat{\bm z},
 \label{eq:tetrahedron}
\end{equation}
where $-1<t<1$. The linear four-spin Hamiltonian
$H(t)=-\sum_{j=0}^3u_j(t)\cdot\Sop_j$ has a unique product ground
state, constant gap $1$, and moments $\m_j=su_j$. Orient the face
$(0,1,2)$ counterclockwise and the remaining faces consistently.
For $t>0$, zero lies inside the resulting tetrahedron, so radial
projection has degree $1$. For $t<0$, all four moments lie in the
southern hemisphere, giving degree $0$. At $t=0$, the first face hull
contains zero. Every local polarization remains exactly one. The chiral
model below realizes the same geometric obstruction in an interacting
ground state under variation of a uniform field.

\section{A physical chiral-spin Hamiltonian}
\label{sec:model}

For the numerical model and its exact certificates, specialize to $s=1/2$
and $\Sop_i=\bm\sigma_i/2$. Take a triangular lattice with primitive vectors
$a_1=a(1,0)$ and $a_2=a(1/2,\sqrt3/2)$. The quantum flake is
\begin{equation}
 \Qsites_R=\{ma_1+na_2:\max(\abs m,\abs n,\abs{m+n})\leq R\},
 \qquad N=1+3R(R+1).
 \label{eq:flake}
\end{equation}
For each oriented nearest-neighbour bond, let
$\bm D_{ij}=D\widehat{\bm z}\times\widehat{\bm e}_{ij}$, where
$\widehat{\bm e}_{ij}$ points from $i$ to $j$. Define
\begin{align}
 H(B)={}&\sum_{\langle ij\rangle\subset\Qsites_R}
 \left[-J\Sop_i\cdot\Sop_j-K\widehat S_i^z\widehat S_j^z
       +\bm D_{ij}\cdot(\Sop_i\times\Sop_j)\right]
 \nonumber\\
 &-B\sum_{i\in\Qsites_R}\widehat S_i^z
   -h_x\widehat S_0^x-h_y\widehat S_0^y+H_{\partial},
 \label{eq:H}\\
 H_{\partial}={}&
 \sum_{\substack{\langle ij\rangle\\i\in\Qsites_R,\ j\notin\Qsites_R}}
 \left[-J\Sop_i\cdot\bm c_j-K\widehat S_i^zc_j^z
       +\bm D_{ij}\cdot(\Sop_i\times\bm c_j)\right],
 \qquad \bm c_j=\tfrac12\widehat{\bm z}.
 \label{eq:boundary}
\end{align}
Every internal bond is counted once; each boundary bond is oriented
from its quantum endpoint to its classical endpoint. The index $0$
denotes the centre. Here $J>0$ is ferromagnetic exchange and $K\geq0$
is optional easy-axis exchange anisotropy. For spin-$1/2$, a single-ion
term proportional to $(\widehat S_i^z)^2$ would instead be constant.
We use $J=a=1$, $D=2$ and $K=0$ throughout the computations. The ratio
$D=2J$ matches the choice of \cite{sotnikov,salvati}, made there so that
the skyrmion size is compatible with a small cluster.

The prescribed exterior selects an upward background without adding
Hilbert-space degrees of freedom. Its use as a finite-flake model is
established in \cite{haller2024}; chiral quantum Heisenberg systems also
support skyrmion textures in larger geometries \cite{haller2022}.
Equations~\eqref{eq:H}--\eqref{eq:boundary} define a Hermitian linear
operator with externally specified coefficients. They describe an
effective spin model, without calibration to a particular material.
The quantum state is obtained by ground-state diagonalization, not by
introducing a state-dependent interpolation.

For the reconstruction, add the classical shell
$\Qsites_{R+1}\setminus\Qsites_R$ and all elementary triangles in the
resulting hexagonal disk. Fill its boundary with a cone to one additional
vertex. Every vertex on the added shell and the cap has moment
$\widehat{\bm z}/2$, so the cap map is constant. This produces a
closed oriented sphere, to which Section~\ref{sec:geometry} applies.
The certificate uses weights one on quantum vertices and zero on the
exterior and cap. The topology therefore refers to the quantum flake
together with its specified magnetic environment.

\begin{figure}[tbp]
 \centering
 \includegraphics[width=\textwidth]{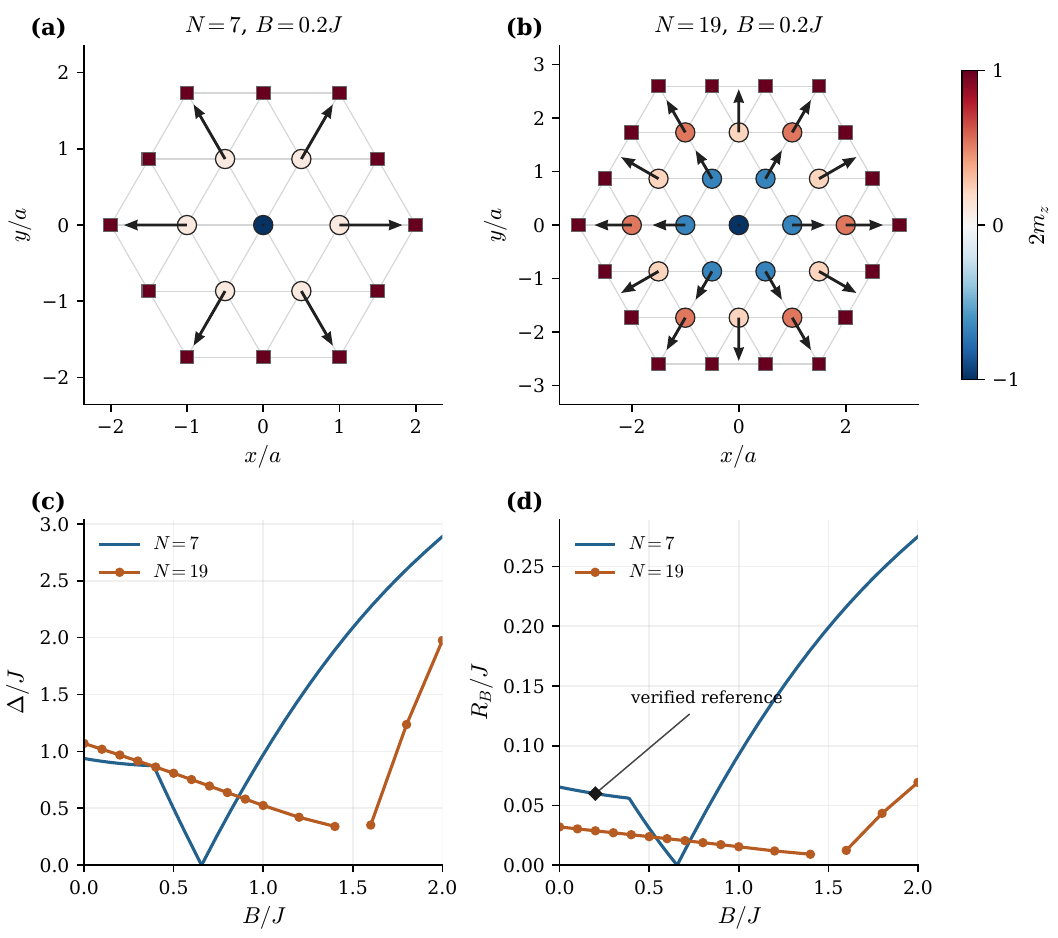}
 \caption{Physical ground-state calculations for $D=2J$, $K=0$ and
 $h_x=h_y=0$. (a),(b) Seven- and nineteen-spin textures at $B=0.2J$;
 colour denotes $2m_z$ and arrows denote in-plane moments. The outer
 ring is the prescribed polarized exterior. (c) Spectral gaps at the
 sampled fields. (d) The sufficient field radius \eqref{eq:fieldradius},
 evaluated from the computed gap and fixed-exterior margin. Lines
 connect samples and do not resolve every level crossing. The marked
 seven-spin reference has the rigorous lower bounds
 \eqref{eq:certifiedreference}; the nineteen-spin values are
 floating-point results.}
 \label{fig:flakes}
\end{figure}

\section{Verified examples and numerical results}
\label{sec:results}

\subsection{Ground-state textures and field robustness}

With $h_x=h_y=0$, both $R=1$ and $R=2$ support charge-$-1$ ground-state
textures at the sampled low fields (figure~\ref{fig:flakes}). They have
nonzero quantum correlations: a pure product state would have every
$p_i=1$, whereas the computed local polarizations are strictly smaller.
For $N=7$ and $B_0=0.2J$, the numerical values are
\begin{equation}
 E_0/J\simeq-9.749030193,\qquad
 \Delta_0/J\simeq0.894236465,
 \qquad \Gamma_0\simeq0.306394206.
 \label{eq:reference}
\end{equation}
The exact-arithmetic verification in Appendix~\ref{app:verification} proves the
conservative bounds
\begin{equation}
 Q=-1,\qquad \Delta_0>0.89423645J,\qquad
 \Gamma_0>0.30639420,
 \qquad R_B(B_0)>0.05992J.
 \label{eq:certifiedreference}
\end{equation}
The original certificate therefore preserves the charge for
$\abs{B/J-0.2}<0.05992$. Corollary~\ref{cor:sharperradius}
improves this certified interval to
\[
 \abs{B/J-0.2}<0.07451.
\]
Both statements hold throughout their respective intervals, rather
than only at sampled fields. The plotted radii in
figures~\ref{fig:flakes} and~\ref{fig:event} use the original
certificate $R_B$ of \eqref{eq:fieldradius}.

For the nineteen-spin flake, the computed values at $B=0$ are
$\Delta/J\simeq1.07049648$, $\Gamma\simeq0.39701934$ and
$R_B/J\simeq0.03202365$. At the same state,
$\gamma\simeq0.37430526$: respecting the fixed exterior improves the
local-data radius by approximately $6.1\%$. The nineteen-spin ground
state has $Q=-1$ at the sampled fields through $B=1.4J$ and is polarized
at $B=1.6J$. This is a finite-flake observation, not a thermodynamic
transition estimate. In the seven-spin system, the crossing to the
polarized state occurs numerically at $B/J\simeq0.6561825841$.
No ground-state charge is assigned at a degeneracy without
specifying a state in its eigenspace.
Independent nineteen-spin reruns with tighter tolerances reproduce the
moments and gaps to within $1.4\times10^{-14}$ and $5\times10^{-13}J$,
respectively, at four selected fields. Appendix~\ref{app:sensitivity}
also quantifies how explicit input-error budgets affect the reported
radii; these diagnostics do not replace exact spectral enclosures.
\FloatBarrier

\subsection{Charge change with a uniformly positive gap}

We now show that the reconstruction can become inadmissible even while
both the spectral gap and every local polarization remain uniformly
positive. Take $N=7$ and the prescribed transverse field
$(h_x,h_y)=(0.2J,0.07J)$. This removes the exact polarized eigenstate
and permits a continuous ground-state texture through the region of
interest. A scan with step $0.0005J$ brackets the numerical charge change between
$0.6480J$ and $0.6485J$. Proposition~\ref{prop:uniqueevent} sharpens
this to a certified bracket of width $2\times10^{-12}$ in $B/J$ containing a
single inadmissible field. We certify the convenient
interval $[0.6480J,0.6490J]$, whose midpoint $0.6485J$ supplies the
reference bounds for both halves. Thus the three fields in
table~\ref{tab:certificates} serve as two endpoints and an interval
reference, not a separate estimate of the transition width.

\begin{table}[htbp]
 \centering
 \caption{Certified quantities for $N=7$, $D=2J$, $K=0$ and
 $(h_x,h_y)=(0.2J,0.07J)$. Displayed bounds are rounded down.
 The charge is an exact signed preimage count followed by the
 residual certificate.}
 \label{tab:certificates}
 \begin{tabular}{@{}ccccc@{}}
 \toprule
 $B/J$ & $Q$ & Lower bound for $\Delta/J$ & Lower bound for $\Gamma$
        & Lower bound for $p_{\min}$\\
 \midrule
 $0.6480$ & $-1$ & $0.0712263$ & $0.00576126$ & $0.5044985$\\
 $0.6485$ & $0$  & $0.0707269$ & $0.00104577$ & $0.4931949$\\
 $0.6490$ & $0$  & $0.0702578$ & $0.00799205$ & $0.4820588$\\
 \bottomrule
 \end{tabular}
\end{table}
\FloatBarrier

\begin{proposition}[A certified geometric obstruction]
\label{prop:physicalexample}
For \eqref{eq:H}--\eqref{eq:boundary} with the preceding parameters,
the ground state is simple throughout
$I=[0.648J,0.649J]$. On that interval,
\begin{equation}
 \Delta(B)>0.0672J,\qquad
 \min_i p_i(B)>0.442.
 \label{eq:uniformbounds}
\end{equation}
Nevertheless, its moment reconstruction is inadmissible for at least
one $B\in I$, with endpoint charges $Q(0.648J)=-1$ and $Q(0.649J)=0$.
\end{proposition}
\begin{proof}
The exact certificates and their verification procedure in Appendix~\ref{app:verification} establish the endpoint degrees and, at
$B_m=0.6485J$, the stronger rational lower bounds
$\Delta_m>0.0707269233J$ and $p_{\min}(B_m)>0.4931949585$.
For $B\in I$, $\norm{H(B)-H(B_m)}\leq0.00175J$.
The min--max principle therefore gives
$\Delta(B)>0.0707269233J-0.0035J>0.0672J$.
Equations~\eqref{eq:rotation} and \eqref{eq:contraction} imply
\begin{equation}
 p_i(B)\geq p_i(B_m)
 -\frac{2(0.00175)}{0.0707269233-0.00175}>0.442.
 \label{eq:polarizationinterval}
\end{equation}
At any $B_0\in I$, choose a positively oriented contour $\mathcal C$
enclosing only $E_0(B_0)$. The positive gap and eigenvalue continuity
ensure that it still isolates the ground eigenvalue for $B$ near $B_0$.
Since $H(B)$ is affine in $B$, its resolvent is analytic there, and
\[
 P(B)=\frac{1}{2\pi\mathrm i}\oint_{\mathcal C}
           (zI-H(B))^{-1}\,\mathrm dz
\]
is the continuous ground-state projector. These local formulas agree
on overlaps, giving continuity throughout $I$. Hence
$\m_i(B)=\Tr(P(B)\Sop_i)$ is continuous. If the reconstruction were
admissible everywhere, homotopy invariance would force equal endpoint
degrees, contradicting the certified values.
\end{proof}

The reconstruction obstruction identified in
Proposition~\ref{prop:physicalexample} can in fact be sharpened to an
exact statement of uniqueness, together with an explanation of its
effect on the plotted certificates in figure~\ref{fig:event}.

\begin{proposition}[Unique transverse event]
\label{prop:uniqueevent}
Consider the seven-spin Hamiltonian with $D/J=2$, $K=0$,
$(h_x/J,h_y/J)=(0.2,0.07)$ and the prescribed exterior
$\widehat{\bm z}/2$. Put $b=B/J$ and $I=[0.648,0.649]$.
The reconstruction uses the mesh of Section~\ref{sec:model}, namely the
complete elementary-triangle disk with its constant cap.
The ground state is simple on $I$, with
\[
 \Delta(b)/J>0.0672,\qquad \min_i p_i(b)>0.442.
\]
There is exactly one inadmissible reconstruction in $I$, at a field
\[
 0.648423852712<b_*<0.648423852714.
\]
The singularity occurs on the oriented face
$\tau_*=(0,1,3)$, whose axial coordinates are
$((0,0),(-1,0),(0,-1))$ in the site ordering of
Appendix~\ref{app:computation}.
At $b_*$ the three moment vectors have rank two and zero lies strictly
inside their triangle. Moreover,
\[
 \frac{d}{db}\det(\m_0(b),\m_1(b),\m_3(b))>1.747
 \quad (b\in I).
\]
Every other variable face has fixed-exterior radius greater than
$0.0155$ throughout $I$. Consequently,
\[
 Q(b)=-1\quad (0.648\leq b<b_*),\qquad
 Q(b)=0\quad (b_*<b\leq0.649).
\]
\end{proposition}

\begin{proof}
All derivatives in this proof are with respect to the dimensionless
field $b$, and the Hamiltonian is divided by $J$.
The midpoint spectral certificate and the endpoint degree certificates
of Proposition~\ref{prop:physicalexample} give a gap greater than
$g=0.0672$, polarization greater than $0.442$, and endpoint degrees
$-1$ and $0$. We first derive a uniform derivative bound that turns
additional point certificates into a global uniqueness statement.

Write $A=H'(b)=-\sum_i\widehat S_i^z$, so that
$L=\norm A=7/2$. Choose a normalized ground eigenvector $\psi(b)$
in horizontal gauge, $\langle\psi,\psi'\rangle=0$, and let
$P=\ket\psi\bra\psi$, $P_\perp=1-P$, and
\[
 \mathcal R=P_\perp(H-E)^{-1}P_\perp.
\]
Here the inverse is taken on the orthogonal complement of the ground
state. Simplicity and the uniform gap imply $\norm{\mathcal R}\leq1/g$.
Differentiating the eigenvalue equation gives
\[
 \psi'=-\mathcal{R}A\psi,\qquad E'=\langle\psi,A\psi\rangle,
 \qquad \norm{\psi'}\leq L/g.
\]
A second differentiation and projection give
\[
 P_\perp\psi''=-2\mathcal R P_\perp(A-E')\psi',\qquad
 P\psi''=-\norm{\psi'}^2\psi.
\]
Therefore $\norm{\psi''}\leq5L^2/g^2$.
For a spin-$s$ local moment, testing against an arbitrary unit direction
and using $\norm{u\cdot\Sop_i}=s$ yields
\[
 \abs{\m_i'}\leq\frac{2sL}{g},\qquad
 \abs{\m_i''}\leq\frac{12sL^2}{g^2}.
\]
For $d(b)=\det(\m_0,\m_1,\m_3)$, the three second-derivative
terms and six mixed first-derivative terms consequently give
\begin{equation}
 \abs{d''(b)}\leq M:=\frac{60s^3L^2}{g^2}
 =\frac{1953125}{96},\qquad s=\tfrac12.
 \label{eq:dsecond}
\end{equation}

At each grid point $b_j=0.648+j/10000$, $0\leq j\leq10$,
the exact-arithmetic verification of Appendix~\ref{app:verification} supplies
a full quantized eigensystem certifying the ground-state ordering and
a trace-distance error $\eta_j$.
The candidate moments are exact rational quadratic expressions in the
quantized ground vector. If $\widehat d_j$ is their determinant,
multilinearity and $\abs{\m_i},\abs{\widehat\m_i}\leq1/2$ imply
\begin{equation}
 \widehat d_j-\tfrac34\eta_j\leq d(b_j)
 \leq\widehat d_j+\tfrac34\eta_j.
 \label{eq:deterror}
\end{equation}
Denote these rational endpoints by $\underline d_j,\overline d_j$.
For a grid cell of length $h=10^{-4}$, integration of
\eqref{eq:dsecond} gives, at every point of the cell,
\begin{equation}
 d'(b)\geq
 \frac{\underline d_{j+1}-\overline d_j}{h}-\frac{Mh}{2}.
 \label{eq:secant}
\end{equation}
Indeed, the difference between the derivative at a point and its
average over the cell is at most
$M h^{-1}\int_{b_j}^{b_{j+1}}\abs{b-t}\,dt\leq Mh/2$.
All ten rational lower bounds exceed $1.747$. Thus $d$ is strictly
increasing on the entire interval, not merely on the sampled grid.

Two further exact certificates give the outward-rounded enclosures
\begin{align*}
 d(0.648423852712)&\in[-3.452,-3.148]\,10^{-12},\\
 d(0.648423852714)&\in[ 2.261, 2.546]\,10^{-12}.
\end{align*}
Continuity and strict monotonicity prove that $d$ has exactly one zero
in $I$, within the stated bracket.

It remains to show that this determinant zero is the only reconstruction
obstruction. At the midpoint, use the archived rational separating
vectors of Appendix~\ref{app:verification} on all variable faces except
$\tau_*$. Let $\ell_\tau$ denote the
resulting candidate lower bound on a face radius, let $\eta_m$ be the
midpoint trace-error bound, and let $\Delta_m$ be its certified gap
lower bound. The norm of the perturbation from the midpoint to any
point of $I$ is at most $\varepsilon=7/4000$. The facewise Lipschitz
bound and the spectral perturbation estimate give
\[
 r_{w,\tau}(b)\geq
 \ell_\tau-\eta_m-\frac{\varepsilon}{\Delta_m-\varepsilon}.
\]
The minimum of the right-hand sides over $\tau\ne\tau_*$ is an
exact positive rational number greater than $0.0155$.
The mesh variant on which these witnesses were computed and the
complete elementary-triangle disk have identical nonconstant face
constraints, apart from repetitions of a single one-variable
constraint, by Appendix~\ref{app:meshrepair}. The bound therefore
applies to the complete disk as well.

The different endpoint degrees force at least one inadmissible field.
All other faces are uniformly admissible, so $\tau_*$ must fail there.
Failure entails $d=0$, whose unique zero has just been proved.
No moment vanishes because of the polarization bound. If zero belonged
to an edge of $\tau_*$, the other face sharing that edge would also
fail. This contradicts its positive radius. A rank-one failed face
with nonzero moments would contain an antipodal pair, causing the same
edge failure. Hence the singular face has rank two and zero lies
strictly in its interior. Admissibility and homotopy invariance on each
side give the stated charges.

The pointwise certificates and rational sign tests used here are
verified in exact arithmetic as described in Appendix~\ref{app:verification};
floating-point eigensolvers supply candidate data only.
\end{proof}

The event's effect on figure~\ref{fig:event} follows a general local
law.

\begin{proposition}[A transverse face event produces a linear cusp]
\label{prop:cusp}
Let a $C^2$ local-moment path have exactly one singular face at $t=t_*$.
Assume that this face $\tau=(i,j,k)$ has three variable vertices of
weight one, its moments have rank two, and zero lies strictly inside
their triangle. Suppose every other face radius is positive at $t_*$,
and put
\[
 d(t)=\det(\m_i(t),\m_j(t),\m_k(t)),\qquad
 A_\tau(t)=\abs{(\m_j-\m_i)\times(\m_k-\m_i)}.
\]
If $d'(t_*)\ne0$, then in a neighborhood of $t_*$,
\begin{align}
 \Gamma(t)&=\frac{\abs{d(t)}}{A_\tau(t)}
 =\kappa\abs{t-t_*}+O((t-t_*)^2),
 &\kappa&=\frac{\abs{d'(t_*)}}{A_\tau(t_*)}>0,
 \label{eq:cusp}\\
 Q(t_*+)-Q(t_*-)&=\operatorname{sgn}d'(t_*).
 \label{eq:jump}
\end{align}
If the spectral gap is positive and Lipschitz at $t_*$, the original
field certificate satisfies
\begin{equation}
 R_B(t)=\frac{\Delta(t_*)\kappa}{2Ns^2}\abs{t-t_*}
       +O((t-t_*)^2).
 \label{eq:radiuscusp}
\end{equation}
When $t=B/J$, use $\Delta/J$ and $R_B/J$ in this last formula.
\end{proposition}

\begin{proof}
Rank two and strict interior containment imply
$A_\tau(t_*)>0$ and strictly positive barycentric coordinates for zero.
The perpendicular projection of zero onto the affine plane of the
triangle depends continuously on its three vertices. Its barycentric
coordinates therefore remain positive near $t_*$, so it remains the
closest point in the triangle, not merely in its affine plane.
The distance-to-plane formula gives
$r_{w,\tau}(t)=\abs{d(t)}/A_\tau(t)$.
Every other face radius is bounded away from zero nearby, while this
one tends to zero. Thus it equals the global radius $\Gamma$ locally.
Taylor's theorem for $d$ and continuity with a first-order bound for
$A_\tau$ prove \eqref{eq:cusp} with the stated remainder.

For the degree jump, normalize the three moments. At the singularity
their directions lie on a great circle and surround its origin without
an antipodal pair. The three successive angular gaps
$\alpha,\beta,\gamma$ lie in $(0,\pi)$ and sum to $2\pi$. Therefore
\[
 a_\tau=1+\cos\alpha+\cos\beta+\cos\gamma
 =-4\cos(\alpha/2)\cos(\beta/2)\cos(\gamma/2)<0.
\]
Also $b_\tau=d/(\abs{\m_i}\abs{\m_j}\abs{\m_k})$.
The principal argument of $a_\tau+\mathrm i b_\tau$ jumps by
$2\pi\operatorname{sgn}d'(t_*)$. The other face areas are continuous,
which proves \eqref{eq:jump} with the manuscript's orientation
convention. Finally, substitute \eqref{eq:cusp} into
$R_B=\Delta\Gamma/[Ns(2s+\Gamma)]$.
A Lipschitz gap is sufficient for the $O((t-t_*)^2)$ remainder;
simplicity of the first excited eigenvalue is not required.
\end{proof}

\begin{remark}[Positive unequal weights]
The same geometric cusp proof holds when the active face has constant
weights $w_i,w_j,w_k>0$: replace its moments by
$\m_i/w_i,\m_j/w_j,\m_k/w_k$ in the determinant and affine-area
formulas. Positive rescaling preserves the normalized directions and
the sign of the degree jump. The other weighted face radii must remain
positive. This statement does not assert the same formula for a face
with a zero weight. The event in
Proposition~\ref{prop:uniqueevent} occurs on a fully quantum face, so
this weight-one case applies directly.
\end{remark}

At the numerical event location, an independent full-eigenbasis
calculation gives
\[
\begin{aligned}
 b_*&\approx0.6484238527131568,
 & d'(b_*)&\approx2.8523372739,\\
 A_{\tau_*}(b_*)&\approx0.2080139046,
 & \kappa&\approx13.71224332.
\end{aligned}
\]
so that $R_B/J\approx0.2773831941\,\abs{b-b_*}$ to leading order near
the event. The event bracket, uniqueness, and transversality are
rigorously certified in Proposition~\ref{prop:uniqueevent}, and the
linear asymptotics follow from Proposition~\ref{prop:cusp}. The
displayed decimal values of $d'(b_*)$, $A_{\tau_*}(b_*)$, $\kappa$, the field-radius coefficient and the chirality values below are numerical estimates without certified error bounds.

The reconstruction certified above uses only the local moments. It is
natural to compare it with the quantum scalar chirality
\begin{equation}
 \widehat\chi_\tau=\Sop_i\cdot(\Sop_j\times\Sop_k),
 \qquad \tau=(i,j,k),
 \label{eq:chirality}
\end{equation}
proposed in \cite{sotnikov} as a many-body diagnostic and studied under
local projective measurements in \cite{salvati}. On the active face,
$\det(\m_0,\m_1,\m_3)$ is the fully decoupled counterpart of
$\langle\widehat\chi_{\tau_*}\rangle$, obtained by replacing each spin
operator by its expectation value. At $b_*$ the determinant vanishes by
construction, whereas $\langle\widehat\chi_{\tau_*}\rangle\approx-0.06500$
with $d\langle\widehat\chi_{\tau_*}\rangle/db\approx1.977$, comparable to
$d'(b_*)$. Over the sampled interval $B/J\in[0.60,0.70]$ the correlator
ranges over $[-0.0979,-0.0049]$ and passes through the event without a
feature, while the ratio
$\det(\m_0,\m_1,\m_3)/\langle\widehat\chi_{\tau_*}\rangle$ falls from
$0.85$ at $B/J=0.60$ to zero at $b_*$ and changes sign beyond it. For
$h_x=h_y=0$ the same correlator drifts only from $-0.1125$ at $B/J=0.2$
to $-0.0989$ at the polarized crossing, and vanishes identically above
it. These are bare per-face expectation values, not the normalized
chirality of \cite{sotnikov}.

\begin{figure}[tbp]
 \centering
 \includegraphics[width=\textwidth]{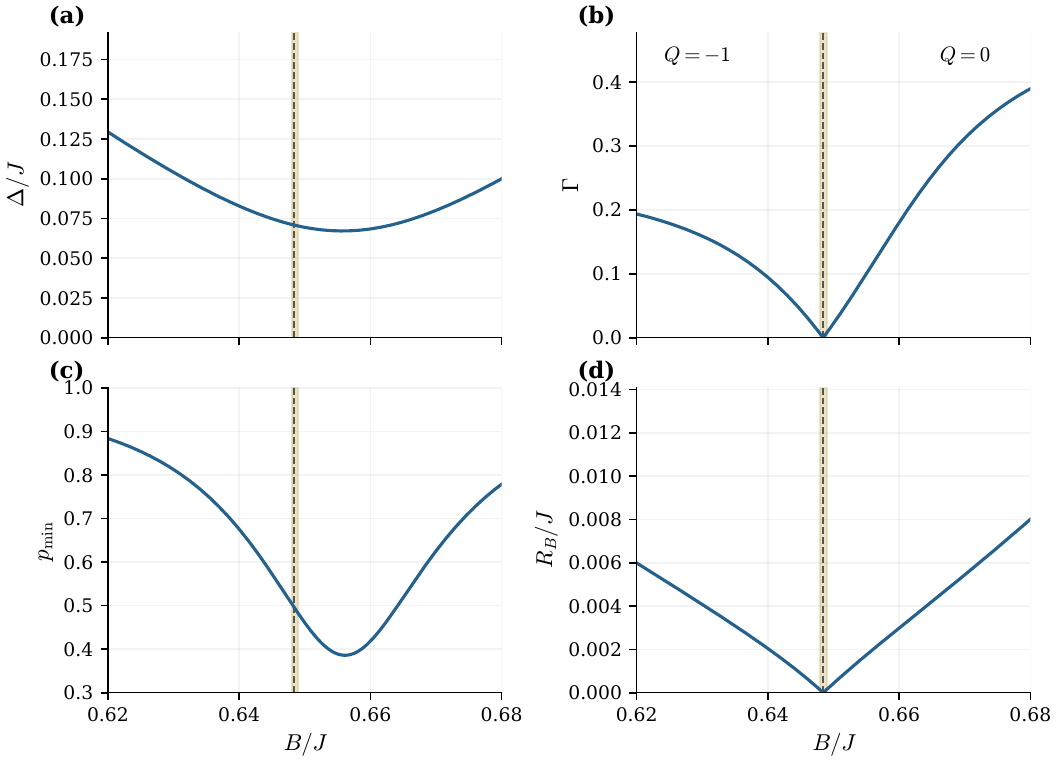}
 \caption{Seven-spin ground states with transverse field
 $(h_x,h_y)=(0.2J,0.07J)$. (a) Spectral gap. (b) Fixed-exterior
 geometric radius. (c) Minimum local polarization. (d) Sufficient
 field radius $R_B$ from \eqref{eq:fieldradius}. The vertical line marks the numerically located
 geometric singularity; the shaded interval
 $[0.648J,0.649J]$ has the rigorous bounds \eqref{eq:uniformbounds}.
 The endpoint charges on this interval are certified to differ.
 For display, the exact zero of $\Gamma$ and $R_B$ at the unique
 event is placed at its numerical location within the certified
 bracket of Proposition~\ref{prop:uniqueevent}; this inserted zero
 is not a floating-point evaluation of the face radius.
 By Proposition~\ref{prop:cusp}, the active-face projection produces a
 linear cusp in $\Gamma$ and in the sufficient field radius, although
 the ground state remains spectrally isolated throughout.}
 \label{fig:event}
\end{figure}
\FloatBarrier

\section{Discussion}

The weighted polyhedral radius is the exact distance to inadmissibility
under the specified local error model. For fixed quantum and classical
regions, it distinguishes uncertainty in measured or computed quantum
moments from prescribed exterior data. Its combination with a spectral
gap gives a sufficient field interval, and the residual formulation
allows an approximate eigenstate to be used without assuming that its
computed moments are exact.

The geometric and spectral conditions have distinct roles. A positive
gap controls the variation of a ground-state vector under a small
operator perturbation. A positive geometric radius ensures that this
variation stays within a single admissible reconstruction sector.
Neither the gap nor nonzero local polarizations alone ensure the latter:
Propositions~\ref{prop:physicalexample} and~\ref{prop:uniqueevent}
establish their simultaneous persistence through a change of
reconstructed degree in the interacting chiral model, at a single
transversely crossed field.

The certified event admits a second reading. The determinant
$\det(\m_0,\m_1,\m_3)$ and the correlator
$\langle\widehat\chi_{\tau_*}\rangle$ of \eqref{eq:chirality} agree only
after the decoupling used in \cite{sotnikov}, which assumes nearly
aligned neighbouring moments. A face degeneracy is precisely where that
assumption fails, and the values reported in
Section~\ref{sec:results} show the two separating: the decoupled product
crosses zero at $b_*$ while the correlator does not, and the latter is
there neither small nor stationary. Its indifference to the event is not
mere insensitivity, since at $h_x=h_y=0$ it does collapse at the
crossing to the polarized state. The obstruction is therefore a property
of the local-moment reconstruction rather than of the state. In this
sense the present results and those of \cite{salvati} are complementary:
there an observable remains stable while the state is strongly altered
by measurement, whereas here the state varies smoothly and remains
gapped while the reconstructed degree changes.

The numerical scope remains finite. The seven-spin statements include
exact-arithmetic verification; the nineteen-spin results are conventional
sparse diagonalization data. Boundary dependence, a thermodynamic limit,
and skyrmion lifetimes require additional analyses. In particular, the
certified field radius is not a tunnelling barrier or a decay time, and
the degree is not a many-body topological-order invariant. Its practical
content is a reproducible bound on when a specified quantum-spin texture
can be assigned an integer charge reliably.

\section*{Conflict of interest}
The authors have no conflicts to disclose.

\section*{Data availability}
The computational supplement is archived at Zenodo,
\doi{10.5281/zenodo.22688715}.
The source code is available in the
\href{https://github.com/raulsanchezgalan/skyrmion-certify}
{project repository}.

\section*{Use of generative AI}
During the preparation of this manuscript, the authors used ChatGPT
as an exploratory and editorial aid and to assist with developing
the numerical simulations.

\appendix

\section{Exact-arithmetic verification}
\label{app:verification}

\subsection{Complete spectral enclosures}

We separate candidate generation from verification, as in validated
numerical computation \cite{rump}. For the seven-spin problem, the
matrix dimension is $128$. With the rational parameters used here,
the Hamiltonian has the form $H=A+\mathrm i(\sqrt3\,C+T)$, with
rational matrices $A,C,T$. For the four certificates of
Section~\ref{sec:results}, an integer-square-root calculation encloses
$\sqrt3$ between adjacent multiples of $2^{-90}$. Rounding each
matrix entry to a multiple of $2^{-48}$ gives a Hermitian matrix
$\widetilde H$ and a rational bound
\begin{equation}
 \norm{H-\widetilde H}\leq\epsilon_H.
 \label{eq:Hround}
\end{equation}
The bound uses the maximum entrywise error times $128$, which bounds the
Frobenius norm. A floating-point eigensolver supplies candidates only.
Their real and imaginary components and the real approximate eigenvalues
are quantized to multiples of $2^{-40}$, giving a square matrix $X$ and
an ordered diagonal matrix $\Lambda$.

The ten additional certificates used in
Proposition~\ref{prop:uniqueevent} enclose $\sqrt3$ at resolution
$2^{-110}$, round Hamiltonian entries to multiples of $2^{-60}$,
and quantize eigenvector components and approximate eigenvalues
to multiples of $2^{-52}$. The verification argument below is the
same for both sets. The exact parameters for every certificate are
recorded in \texttt{data/manifest.json} within the supplement.

All entries of $X^*X-I$ and $\widetilde H X-X\Lambda$ are then computed
by integer matrix multiplication with their exact powers-of-two
denominators. Integer-square-root upper bounds give rational numbers
\begin{equation}
 \delta\geq\norm{X^*X-I}_F,\quad
 r_X\geq\norm{\widetilde H X-X\Lambda}_F,\quad
 M_H\geq\norm{\widetilde H}_F,
 \qquad \delta<1.
 \label{eq:enclosureinputs}
\end{equation}
Let $U$ be the unitary polar factor of $X$. Each singular value
$\sigma$ of $X$ satisfies $\abs{\sigma^2-1}\leq\delta$, so
$\norm{U-X}\leq\delta$. Therefore
\begin{equation}
 \norm{H-U\Lambda U^*}
 \leq\epsilon_H+r_X+(M_H+\norm\Lambda)\delta
 =:\epsilon_{\rm spec}.
 \label{eq:spectralenclosure}
\end{equation}
Indeed, subtract $U\Lambda$ from $\widetilde H U$, insert $X$, and
use \eqref{eq:enclosureinputs}. The min--max principle now encloses
\emph{every ordered eigenvalue}:
\begin{equation}
 E_j(H)\in[\Lambda_j-\epsilon_{\rm spec},
           \Lambda_j+\epsilon_{\rm spec}].
 \label{eq:alllevels}
\end{equation}
This checks ground-state ordering, rather than merely the residuals of
a few candidate eigenpairs.

The first column $x$ gives $\phi=x/\norm x$. Its residual at
$\lambda=\Lambda_0$ is bounded by its exact integer residual divided
by a rational lower bound for $\norm x$, plus $\epsilon_H$.
Taking $b=\Lambda_1-\epsilon_{\rm spec}$ in
Corollary~\ref{prop:residual} yields a rational upper bound $\eta$
for the ground-state trace error. The four certificates of
Section~\ref{sec:results} have $\epsilon_{\rm spec}<10^{-8}J$ and
$\eta<10^{-9}$; the ten certificates of
Proposition~\ref{prop:uniqueevent}, at the finer quantization above,
have $\epsilon_{\rm spec}<2\times10^{-12}J$ and
$\eta<3\times10^{-13}$.

\subsection{Moment bounds and an exact degree count}

Although $\phi$ need not have rational components, its moments are
rational: each is a quadratic expression in $x$ divided by $x^*x$.
Candidate separating vectors from \eqref{eq:dual} are contracted
slightly, quantized to multiples of $2^{-40}$, and checked exactly for
$\abs v\leq1$ and, when necessary, $v_z\geq0$. Their rational dot
products with the candidate moments give a lower bound for
$\Gamma(\phi)$. Subtracting $\eta$ gives the reported bound for the
exact ground state. These separating vectors are stored with the
certificates, and the independent replay of
Appendix~\ref{app:computation} reads them and checks their feasibility
directly, so reproducing a certificate never requires recomputing an
optimizer.
Likewise, if $\ell_i$ is a rational lower bound for $\abs{\m_i(\phi)}$,
the trace-error certificate gives, for the spin-$1/2$ system,
\[
 \abs{\m_i(\psi_0)}\geq\ell_i-\eta,\qquad
 p_i(\psi_0)\geq2(\ell_i-\eta),
\]
so that $p_{\min}(\psi_0)\geq2(\min_i\ell_i-\eta)$.

The integer degree is verified independently of floating-point solid
angles. Choose the ray $u=(137,223,317)$ and form the rational face
matrix $M_\tau=(\m_i\ \m_j\ \m_k)$. When $\det M_\tau\ne0$, the
normalized ray lies in the interior of the radial face exactly when
every component of $M_\tau^{-1}u$ is positive. Its contribution to
the degree is $\operatorname{sgn}\det M_\tau$. Cramer's rule reduces
these tests to integer determinant signs after clearing denominators.
The code verifies that no tested preimage lies on an edge and that the
ray avoids the image of every rank-deficient face. The signed count is
therefore an exact degree of the candidate texture. Corollary
\ref{prop:residual} transfers it to the exact ground state.

The supplement retains the quantized candidates and the rational
bounds above. Every acceptance decision uses Python integers and
fractions, so the conclusions do not depend on rounding properties of
the eigensolver. The verification still assumes that the stated
Hamiltonian and mesh encode the intended finite physical model.
\section{Numerical conventions and reproducibility}
\label{app:computation}

\subsection{Assembly and eigensolvers}

Site $i$ is bit $i$ of the computational basis, with bit zero denoting
spin up. Quantum sites are ordered first by hexagonal shell and then
lexicographically in $(m,n)$. The seven-spin Hamiltonian is also
assembled independently from tensor products of Pauli matrices and
compared with the bit-transition implementation. Mesh checks verify
oppositely oriented incidences at every edge, Euler characteristic two,
and the elementary-triangle inventory: the complete disks contain $24$
and $54$ elementary triangles with $12$ and $18$ boundary edges for
$N=7$ and $N=19$, hence $36$ and $72$ faces after capping. The inventory
is the binding check, since the two enumerations of
Appendix~\ref{app:meshrepair} both give $20$ vertices, $54$ edges and
$36$ faces for $N=7$, and therefore agree on the first two. The exact-arithmetic matrix
provides a further independent check of the seven-spin operator
coefficients.

Seven-spin field scans use complete-Hilbert-space dense Hermitian
diagonalization. The nineteen-spin dimension is $524288$; a sparse
Hermitian eigensolver computes the four lowest states, with random
components included in its starting vector to avoid restriction to an
invariant polarized subspace. Independent starts check selected
nineteen-spin results. The largest reported eigenpair residual on the
nineteen-spin scan is below $4\times10^{-11}J$. These larger-system
values are not supplied with exact spectral-ordering enclosures.

\subsection{Two enumerations of the reconstruction disk}
\label{app:meshrepair}

Section~\ref{sec:model} uses the complete disk of elementary triangles.
Enumerating those triangles from anchors inside the disk alone, as in
the first release of the supplement, omits two boundary triangles for
$N=7$ and three for $N=19$, because a downward-pointing triangle can lie
inside the disk when its anchor does not. We show that the two
enumerations give the same admissibility, the same radii and the same
degree, so no reported quantity depends on the choice.

Each omitted triangle has one quantum vertex and two prescribed exterior
vertices. Restoring it replaces two nonconstant cap faces in the affected
patch by one disk face and one constant cap face. Writing $m$ for the
quantum moment and $c=\widehat{\bm z}/2$, each affected nonconstant face
has
\[
 \conv\{m,c,c\}=\conv\{m,c\},\qquad
 P=m+\cone\{c\}.
\]
Thus the complete enumeration only changes the multiplicity of an
identical fixed-exterior constraint and adds constant faces;
admissibility and $\Gamma$ are unchanged. It also preserves $\gamma$: a
constant face has distance $\abs c=1/2$, which cannot lower the global
radius because $\abs m\leq1/2$ at every quantum vertex.
Whenever admissible, each affected nonconstant face maps into a
single great-circle arc and has zero spherical area. Every face
that can contribute nonzero area is unchanged, so the degree is
preserved. The Hamiltonian does not refer to the reconstruction mesh at
all; its eigenstates, gap, and the resulting field certificates are
therefore unaffected.

The complete disks contain $24$ and $54$ elementary triangles for
$N=7$ and $N=19$, respectively, with $12$ and $18$ cap faces.
Both statements are checked in the supplement: the stored certificates
reproduce their recorded rational bounds exactly when the moment
geometry is recomputed on the complete disk, and the independent
verifier constructs both enumerations, confirms that their nonconstant
face constraints coincide, and evaluates the endpoint degrees on each.

\subsection{Geometry and field refinement}

The face geometry is evaluated by vertex, edge and interior active-set
calculations. For faces adjoining the fixed exterior, the relevant
polyhedron is the convex hull of the variable moments plus the upward
ray. For the transverse-field example, the saved scan covers
$B/J\in[0.60,0.70]$ in steps of $0.0005$. The adjacent samples with
different charges are $0.6480$ and $0.6485$; both identify $(0,1,3)$
as the face of smallest radius. The script
\texttt{sensitivity\_analysis.py} refines the zero of
$\det(\m_0,\m_1,\m_3)$ on this bracket using Brent's method with
absolute field tolerance $5\times10^{-15}J$, and records the endpoint
determinants and the positive barycentric coefficients at the root.
The event calculation in \texttt{run\_analysis.py} uses the wider
bracket $[0.648,0.649]J$ and agrees at the displayed precision. Neither
solver tolerance is an error certificate for the true root: the
existence proof uses the exact endpoint charges and uniform interval
bounds instead. The narrower bracket and uniqueness in
Proposition~\ref{prop:uniqueevent} follow from the additional exact
determinant enclosures and the derivative estimate \eqref{eq:secant}.

\subsection{Sensitivity of the nineteen-spin values}
\label{app:sensitivity}

At $B/J=0,0.4,1.4,2$, additional independent random starts request six
eigenpairs with tolerance $2\times10^{-13}$, compared with four and
$2\times10^{-12}$ in the main scan. The maximum discrepancies are
$1.4\times10^{-14}$ in the Euclidean norm of any local moment,
$5\times10^{-13}J$ in the gap and $2\times10^{-14}J$ in $R_B$;
all four charges agree. These are empirical comparisons between solver
runs, not bounds on the unknown exact errors.

To quantify sensitivity separately, suppose the computed data
$(\widehat\m,\widehat\Delta)$ have errors at most
$\epsilon_m$ in the maximum local-moment norm and
$\epsilon_\Delta$ in the gap, with
$\widehat\Gamma=\Gamma(\widehat\m)$. Define
$G_-=\widehat\Gamma-\epsilon_m$,
$G_+=\min\{1/2,\widehat\Gamma+\epsilon_m\}$ and
$D_\pm=\widehat\Delta\pm\epsilon_\Delta$, assuming the lower
endpoints are positive. Lipschitz continuity of $\Gamma$ and
monotonicity of \eqref{eq:fieldradius} give the conditional bounds
\begin{equation}
 \Gamma\geq G_-,\qquad
 p_{\min}\geq\widehat p_{\min}-2\epsilon_m,\qquad
 \frac{2D_-G_-}{19(1+G_-)}\leq R_B
 \leq\frac{2D_+G_+}{19(1+G_+)}.
 \label{eq:sensitivity}
\end{equation}
For the illustrative budgets $\epsilon_m=10^{-6}$ and
$\epsilon_\Delta=10^{-6}J$, the interval for $R_B/J$ is contained
in $[0.03202356,0.03202375]$ at $B=0$ and
$[0.00914571,0.00914581]$ at $B=1.4J$. All sampled textures have
$\widehat\Gamma>0.3428$, well above the moment budget. These round
budgets are chosen to display sensitivity and are not estimated
confidence limits or verified error bounds. In particular, this
calculation assumes correct ground-state identification and cannot
exclude a missed lower eigenvalue. The supplement records the full
conditional envelopes and the independent-run comparisons.

\subsection{Replaying the computations}

The supplement records software versions and solver settings. Its
notebook recreates the figures from the saved data and reruns the
seven-spin analysis, the exact certificates, the field refinement and
the sensitivity envelopes, exposing the more expensive nineteen-spin
scan and its independent reruns as explicit options.

Two independent replays are provided, neither of which reruns an
eigensolver or regenerates separating vectors. Running
\texttt{python check\_saved\_certificates.py} recomputes the moment
geometry from the model code and checks the four reference certificates
against their recorded rational bounds. Running
\texttt{python replay\_certificates.py} checks all fourteen stored
candidates together with the rational inequalities for the sharper
reference radius and for the unique-event certificate; it reconstructs
the Hamiltonian from tensor-product Pauli operators and does not import
the model code, so it shares no implementation with the analysis it
checks.

\end{document}